\documentclass[11pt]{article}

\usepackage{hyperref}
\hypersetup{
    colorlinks=true,
    linkcolor=blue,
    filecolor=magenta,      
   urlcolor=blue,
}

\usepackage{color}
\usepackage{latexsym}
\usepackage{epsfig,graphics}
\usepackage{amsmath,amssymb,amsthm}
\usepackage{graphics,epsfig,calc}
\usepackage{bm}

\usepackage{upref}
\usepackage{setspace}
 \newcommand{\beqn}{\begin{eqnarray}}
 \newcommand{\eeqn}{\end{eqnarray}}
 \newcommand{\be}{\begin{equation}}
 \newcommand{\ee}{\end{equation}}
 \newcommand{\ba}{\begin{array}}
 \newcommand{\ea}{\end{array}}

 \newcommand{\pa}{\partial}
\newcommand{\la}{\label}

\newcommand{\rRe}{{\rm Re\5}}

 \newcommand{\fr}{\frac}
 
\newcommand{\ov}{\overline}
\newcommand{\ti}{\tilde}

\newcommand{\cF}{{\cal F}}
\newcommand{\cE}{{\cal E}}
\newcommand{\ve}{\varepsilon}

\newcommand{\De}{\Delta}

\newcommand{\al}{\alpha}

\newcommand{\si}{\sigma}
\newcommand{\om}{\omega}

\newcommand{\na}{\nabla}
\newcommand{\lam}{\lambda}
\newcommand{\Lam}{\Lambda}

\newcommand{\5}{{\hspace{0.5mm}}}

\newcommand{\R}{\mathbb{R}}

\newcommand{\C}{\mathbb{C}}

\newtheorem{theorem}{Theorem}[section]

\newtheorem{defin}[theorem]{Definition}

\newtheorem{lemma}[theorem]{Lemma}
\newtheorem{remark}[theorem]{Remark}

\newtheorem{cor}[theorem]{Corollary}
\newtheorem{pro}[theorem]{Proposition}
\newcommand{\bp}{\begin{pro}}
\newcommand{\ep}{\end{pro}}

\newcommand{\bt}{\begin{theorem}}
\newcommand{\et}{\end{theorem}}

\newcommand{\bl}{\begin{lemma}}
\newcommand{\el}{\end{lemma}}
\newcommand{\const}{\mathop{\rm const}\nolimits}

\newcommand{\bpr}{\begin{proof}}
\newcommand{\epr}{\end{proof}}
\newcommand{\br}{\begin{remark}}
\newcommand{\er}{\end{remark}}
\newcommand{\bd}{\begin{defin}}
\newcommand{\ed}{\end{defin}}

\newcommand{\bc}{\begin{cor}}
\newcommand{\ec}{\end{cor}}

\begin{document}
\begin{center}
{\Large\bf Global attraction to solitons for 2D Maxwell--Lorentz 
 \bigskip
 
equations with spinning particle}
 \bigskip \bigskip

 {\large E.A. Kopylova}\footnote{ 
 Supported partly by Austrian Science Fund (FWF) P34177
 }
 \medskip
 \\
{\it
Faculty of Mathematics of   Vienna University\\
 }
 
\medskip

 {\large A.I. Komech} \medskip
 \\
{\it
 Institute of Mathematics, BOKU University, Vienna
}
 
\end{center}
\begin{abstract}
We consider the 2D Maxwell--Lorentz system which describes  a rotating  particle in electromagnetic field. 
The system admits stationary soliton-type solutions.
We prove the attraction to solitons for any finite energy solution relying on the conservation of  angular momentum. 
 \end{abstract}
\setcounter{equation}{0}
\section{Introduction} 
The 2D Maxwell--Lorentz equations reads \cite{KK2023, S2004}
\be\la{mls0}
\left\{\ba{rcl}
\dot E(x,t)&=&J\na B(x,t)-[\dot q(t)-\om(t) J(x-q(t))] \rho(x-q(t))\\
\dot B(x,t)&=&-\na \cdot(JE(x,t)),
\qquad
\na\cdot E(x,t)=\rho (x-q(t))\\
m\ddot q(t)&=&\langle E(x,t)+B(x,t) [J\dot q(t) +\om(t) (x-q(t)) ],\rho(x\!-\!q(t))\rangle\\
I\dot \om(t)&=&\langle (x\!-\!q(t))\cdot \big[J E(x,t)\!-B(x,t) \dot q(t)],\rho(x\!-\!q(t))\rangle
\ea\right|, \quad x\in\R^2. 
\ee
Here $(E(x, t), B(x, t))\in \R^2\oplus\R^2$ is the Maxwell field, $\rho(x-q(t))\in\R$ is the charge distribution of
an extended particle centered at a point $q(t)\in R^2$, $m$ is the mass of the particle, $I$ is its moment of inertia, 
and $\omega(t)$ is the angular velocity of the particle rotation.
The brackets $\langle\,,\,\rangle$ denote the inner product in the  Hilbert space $L^2:=L^2(\R^2)\otimes\R^2$,
and 
\be\la{J}
J=\begin{pmatrix} 0& 1\\ -1& 0\end{pmatrix}.
\ee
We  restrict ourselves  to  the situation, where  the spinning particle is located at the origin, i.e. $q(t)\equiv 0$. 
This can be achieved by assuming the (anti-) symmetry conditions
$E(-x)=-E(x)$, $B(-x)=B(x)$  for the initial fields. Then this property persists for all times.
In this case the Maxwell--Lorentz equations \eqref{mls0} simplify to the linear system
\be\la{mls2}
\left\{\ba{rcl}
\dot E(x,t)&=&J\na B(x,t)+\om(t) Jx\rho(x),
\\
\dot B(x,t)&=&-\na \cdot JE(x,t), \qquad \na\cdot E(x,t)=\rho (x)\\
I\dot \om(t)&=&\langle x\cdot J E(x,t),\rho(x)\rangle
\ea\right|.
\ee
We assume that   real-valued charge density $\rho(x)$ is a smooth rapidly decreasing   spherically symmetric function, i.e.,
\be\la{rosym}
\rho\in  {\cal S}(\R^2),\qquad\rho(x)=\rho_{1}(|x|).
\ee
Moreover, we impose that the total charge vanishes (neutrality of the particle):
\be\la{zero}
\int \rho(x)\, dx=0.
\ee

The system admits stationary solutions (solitons) $(E_{\om}(x), B_{\om}(x),\om)$ 
rotating with constant angular velocity  $\om\in\R$.
The  condition $\hat\rho(0)=0$  provides that the solitons belong to the finite energy phase space (\ref{psp}).

One main result is the long-time convergence in  weighed energy norms of any finite energy solution with the angular momentum $M$
to the soliton  with the  same angular momentum. 

Our proof  relies on the  regularity of the  resolvent $({\bf H}-\lam)^{-1}$ at the threshold $\lam=0$ , where  
\be\la{AA2}
{\bf H}:=\left(
\ba{cccc}
0 & 1 \\
\De -\frac {1}{I}|J\varrho\rangle \langle J\varrho| & 0 \\
\ea\right),\quad \varrho (x)=x\rho(x)
\ee
is the operator defined on the space   \eqref{psp}.
We calculate the Puiseux expansion of the resolvent up to $\lam^2$ by suitable development of the 
Agmon--Jensen--Kato theory \cite{A, JK1979, KK2012, K2010}.     
The condition \eqref{zero} provides that  the resolvent belongs  to $C^2$ at a vicinity of  $\lam=0$. 
This fact implies the decay $|t|^{-2}$ of the solution to (\ref{lin}) in weighted energy norms (\ref{Ebet}) for sufficiently large moment of inertia $I$. 
This decay rate allows us to obtain scattering asymptotics  \eqref{AP-as} in the global energy  norm (without a weight).

Note that  the decay $|t|^{-2}$ in the weighted energy norms (\ref{Ebet}) does not contradict the well known slow local decay $|t|^{-1}$ of 
solutions of the 2D wave equation. 
Indeed, the slow decay  $|t|^{-1}$ corresponds to the singularity 
of the resolvent $(\De-\lam)^{-1}\sim \log\lam$ at the threshold $\lam=0$.  
The corresponding {\it resonances} are the constant functions.
However,  the decay  in the weighted norms  (\ref{Ebet}) 
is better since the  norm includes spatial derivatives of the 
solution (which vanish on the resonances), but does not include the values of the solution.

Let us comment on  previous related results. 
The global attraction to stationary states was established in \cite{KKS2018, KS1998} for 3D wave  system 
with a moving  particle  without rotation. 
In \cite{KS2000}, the global attraction to stationary states was established  for 3D Maxwell--Lorentz system 
with a relativistic particle without rotation in presence of an external confining potential. 
The global convergence to solitons of such system without  external  potential was proved in \cite{IKM2004}.
The results \cite{IKM2004, KS2000} provide the first rigorous proof of the {\it radiation damping} in classical electrodynamics. 
These results were obtained under the Wiener-type condition on the charge density $\rho(x)$.

The global convergence to rotating solitons was established in \cite{IKS2004} for solutions to the 3D Maxwell--Lorentz system  with $q(t)\equiv 0$
in the case of sufficiently small charge density $\rho$.
The global attraction to the set of all solitons  for this system
was established in \cite{K2010}  under a considerably weaker Wiener-type condition than in \cite{IKM2004}. 
{The survey of the results for 3D system can be found in} \cite{KK2020, KK2022}.

In \cite{KK2023}, we have proved the orbital stability of moving and rotating solitons for the 2D Maxwell--Lorentz system \eqref{mls0}.
The proof relies on the reduction of the system by the canonical transformation  to a comoving frame 
and the  conservation of the corresponding linear and angular momenta. 
\setcounter{equation}{0}
\section{The Maxwell potentials}
In this section we recall the Hamiltonian
structure for  the system (\ref{mls2}) expressed in the Maxwell potentials $A(x,t)=(A_1(x,t),A_2(x,t))$ and $\Phi(x,t)$:
\be\la{AA}
B(x,t)=\na\cdot (JA)=\na_1A_2(x,t)-\na_2 A_1(x,t),\qquad E(x,t)=-\dot A(x,t)-\na \Phi(x,t).
\ee
We choose the Coulomb gauge
\be\la{Cg}
\na\cdot A(x,t)=0.
\ee
Then the  first two equations of  (\ref{mls2}) are equivalent to the system
\be\la{2mA}
\left\{\ba{rcl}
-\ddot A(x,t)-\na\dot\Phi(x,t)&=&-\De A(x,t)+\om(t) Jx\rho(x), \\
-\De\Phi(x,t)&=&\rho(x)
\ea\right|.
\ee
The second equation  of (\ref{2mA}) can be solved explicitly:
\be\la{2mA3}
\Phi(x,t)=\Phi(x)=-\fr1{2\pi}\int \log |x-y| \rho(y)dy.
\ee
In the Fourier transform $\hat\Phi_0(k)=\hat\rho(k)/k^2$.
Hence, $\Phi(x)\in {\cal S}(\R^2)$ due to \eqref{zero} and spherical symmetry (\ref{rosym}).
Now the last equation of (\ref{mls2}) becomes
\be\la{laq}
I\dot \om(t)=-\langle J(\dot A(x,t)+\na\Phi(x),\varrho(x)\rangle
=-\langle J\dot A(x,t),\varrho(x)\rangle,\quad \varrho(x):=x\rho(x)
\ee
since the spherical symmetry  (\ref{rosym}) implies
$$
\langle J\na\Phi(x),\varrho(x)\rangle=\langle x\cdot J\na\Phi(x),\rho(x)\rangle=-\langle \na\cdot Jk \hat\Phi(k),\hat\rho(k)\rangle
=-\int (k_2\na_1-k_1\na_2)\frac{\hat\rho^2(k)}{k^2} dk=0.
$$
The angular momentum is defined as in \cite{KK2023}:
\be\la{pM}
M=I\om-\langle A,J\varrho\rangle.
\ee
By \eqref{laq} and \eqref{AA},
$$
\dot M(t)=I\dot\om(t) -\langle \dot A(x,t),J\varrho(x)\rangle\\
\nonumber
=-\langle J \dot A,\varrho(x)\rangle-\langle \dot A,J\varrho(x)\rangle=0.
$$
Thus, the angular momentum is conserved:
$M(t)=M$,  $t\in\R$.
Hence,
the system  (\ref{mls2}) can be rewritten as the linear inhomogeneous system
\be\la{mls3}
\left\{\ba{llll}
\dot A(x,t)=\Pi(x,t)
\\
\dot \Pi(x,t)=\Delta A(x,t)-\om(t) J\varrho(x)
\ea\right|,
\ee
where
\be\la{rw}
\om(t):=\frac{1}{I}\big(M+\langle A(t),J\varrho\rangle\big).
\ee
Let us introduce a phase space for the system \eqref{mls3}. Denote by 
$\dot H^1$  the closure of $C_0^\infty (\R^2)\otimes \R^2$ with respect to the norm 
$\Vert A\Vert_{\dot H^1}=\Vert \nabla A\Vert_{L^2}$, where $L^2:=L^2(\R^2)\otimes\R^2$.
Define the energy phase space
\be\la{psp}
{\cal E}=\dot H^1\oplus L^2.
\ee
Obviously, (\ref{mls3}) is the canonical Hamilton system with the Hamilton functional
\be\la{Hp}
H(Y)=\fr12\int\big(\Pi^2+|\na A|^2\big)dx+\fr{I\om^2}2,\quad Y:=(A,\Pi)\in {\cal E},
\ee
where $\frac{I\om^2}2=\frac{1}{2I}(M+\langle A(x),J\varrho(x)\rangle)^2$ in accordance with \eqref{rw}.
The Hamiltonian (\ref{Hp}) is well defined and  Fr\'echet  differentiable on the Hilbert phase space (\ref{psp}).
The well-posedness  for the system \eqref{mls3} is proved in 
\cite{KK2023}.

\begin{pro} \la{wp}
Let \eqref{rosym} holds, and let $Y_0=(A_0, \Pi_0)\in {\cal E}$. Then\\
(i) there exists a unique solution $Y(t)\in C(\R, {\cal E})$ to the Cauchy problem for \eqref{mls3};\\
(ii) the energy is conserved: $H(Y(t))=H(Y_0)$ for  $t\in R$;\\
\end{pro}
\subsection{Solitons}
The solitons of the system (\ref{mls3}) are stationary solutions 
\be\la{solvom}
Y_{\omega}=(A_{\om}(x), 0).
\ee
Substituting  into \eqref{mls3}, we get
\be\la{solit2}
\Delta  A_{\omega}(x)=\omega J\varrho(x).
\ee
In the Fourier representation
\be\la{solit3}
\hat A_{\omega}(k)=-\fr{\omega J\hat\varrho(k)}{k^2}.
\ee
By \eqref{zero},
$\hat A_{\om}(k) \sim k$,   as $|k|\sim 0$. 
Hence, $\na A_{\om}\in L^2(\R^2)$.
The angular momentum $M_{\om}$ of the soliton (\ref{solvom})
is expressed by
\be\la{mosol}
M_{\om}=I\om+\om\langle\fr{J\ti\varrho}{k^2}, J\ti\varrho\rangle=\om(I+\varkappa_0),\quad\varkappa_0=\int\frac{|\hat\varrho(k)|^2dk}{k^2}.
\ee
\setcounter{equation}{0}
\section{Main result}
To state our main result we need some notations.
Denote
\be\la{kappa} 
\varkappa(\lam)=\int\frac{|\hat\varrho(k)|^2dk}{k^2+\lam^2},\quad\rRe\lam>0.
\ee
The function $\varkappa(\lam)$ is analytic in the halfplane $\rRe\lam>0$.
By \eqref{rosym} and by the Sokhotski-Plemelj formulas, there exists the limit
$$
\varkappa(i\mu+0)=\lim_{\ve\to +0}\varkappa(i\mu+\ve),\quad \mu\in\R,
$$
where $\varkappa(+0)=\varkappa(0)$ by (\ref{zero}).
By Corollary \ref{cor-as},  $\varkappa(i\mu+0)\in C^2(\R\setminus 0)\cup C^1(\R)$. We will suppose that 
\be\la{M-condition}
I+\varkappa(i\mu+0)\ne 0,\qquad\mu\in\R.
\ee
Note that the condition holds  obviously for $\mu=0$. It holds for all $\mu\in\R$ if the moment of inertia $I$ is
sufficiently large.

Let $H^s_\beta=H^s_\beta(\R^2)$ be the weighted Sobolev spaces with the finite norm
$$
\Vert f\Vert_{H^s_\beta}=\Vert \langle x\rangle^{\beta}\langle \na\rangle^s f\Vert_{L^2(\R^2)},\quad \langle x\rangle=(1+|x|^2)^{1/2}.
$$
Denote $L^2_{\beta}=H^0_{\beta}$.
\begin{defin}
{\it i)} ${\cal E}_{\beta}$, $~\beta\in R$,  
is the Hilbert space of   vector  fields $(A,\Pi)$ with finite norm
\be\la{Ebet}
\Vert (A,\Pi)\Vert_{\beta}:=\Vert (A,\Pi)\Vert_{{\cal E}_\beta}=\Vert \na A\Vert_{{\bf L}^2_{\beta}}+\Vert \Pi\Vert_{{\bf L}^2_{\beta}}.
\ee
{\it ii)} ${\cal E}_{\beta}^+$ is the Hilbert space of   vector  fields $(A,\Pi)$ with finite norm
$$
\Vert (A,\Pi)\Vert_{{\cal E}_\beta^+}=\Vert A\Vert_{{\bf L}^2_{\beta}}+\Vert \na A\Vert_{{\bf L}^2_{\beta}}+\Vert \Pi\Vert_{{\bf L}^2_{\beta}}.
$$
\end{defin}
In these notations,  ${\cal E}={\cal E}_0$.
\begin{theorem}\la{main}
Let conditions \eqref{rosym},  \eqref{zero} and \eqref{M-condition} hold, and  $Y_0\in {\cal E}_{\beta}^+$ with $\beta >5/2$.  Let 
$Y(t)\in C(\R, {\cal E})$ be the solution to \eqref{mls3} with initial data $Y_0$
and the corresponding angular momentum $M$. 
{\it i)} Then the solution converges to the soliton (\ref{solvom}) with the 
same angular momentum $M$: 
\be\la{atr}
|\om(t)-\om_*|={\cal O}(|t|^{-2}), \quad \Vert Y(t)-Y_{\om_*}\Vert_{-\beta}={\cal O}(|t|^{-2}),\quad t\to\infty,
\ee
where the  frequency $\om_*$ is defined by (\ref{mosol}):
\be\la{om0} 
\om_*:=\fr {M}{I+\varkappa(0)}.
\ee
{\it ii)}
The scattering asymptotics hold
\be\la{AP-as}
Y(x,t)=Y_{\om_*}(x)+W(t)\Psi+r(x,t),\quad t\to\pm\infty.
\ee
Here $W(t)$ is the dynamical group of the free wave equation, $\Psi\in {\cal E}$, 
and the remainder $r(x,t)$ decays to zero in the global energy norm:
\be\la{r-as}
\Vert r(t)\Vert_{\cal E}={\cal O}(|t|^{-1}),\quad t\to\pm\infty.
\ee
\end{theorem}
\setcounter{equation}{0}
\section{Extraction of a soliton}
 We split  a solution to the system \eqref{mls3} as  the sum
\be\la{YZ}
Y(t)=Y_{\om_*}+Z(t),\qquad
Y_{\om_*}=\left(\!\!\ba{c}A_{\om_*}(x)\\ 0\ea\!\!\right),
\quad
Z(t)=\left(\!\!\ba{c}\Lam(x)\\ \Pi(x)\ea\!\!\right).
\ee
Substituting into \eqref{mls3}, we obtain 
\be\la{dot-APi}
\left\{\ba{rcl}
\dot\Lambda(x,t)&=&\Pi(x,t)\\
\dot\Pi(x,t)
&=&\De A_{\om_*}(x)+\De\Lambda(x,t)-\om(t) J\varrho(x)
\ea\right|.
\ee
Using   equation (\ref{solit2}) with $\om=\om_*$, 
and equation (\ref{rw}) with $\om=\om(t)$, we get
\be\la{dotPsi}
\dot \Pi=\om_* J\varrho+\De\Lambda-\om J\varrho=\De\Lambda
-\frac{1}{I}(\langle \Lam+A_{\om_*},J\varrho\rangle-\langle A_{\om_*},J\varrho\rangle)J\varrho
=\De\Lambda-\frac{1}{I}\langle \Lam,J\varrho\rangle J\varrho.
\ee
Now the system \eqref{dot-APi} reads
\be\la{lin}
\dot Z(t)={\bf H}Z(t),\quad t\in\R,
\ee
where the operator ${\bf H}$ is defined by \eqref{AA2}.
\begin{lemma} \la{haml}
i)  The system (\ref{lin}) formally can be written as the Hamilton system
\be\la{lineh}
\dot Z(t)=
JD{\cal H}(Z(t)),~~~~~~~t\in\R,
\ee
where $D{\cal H}$ is the Fr\'echet derivative of the Hamilton functional
\be\la{H0}
{\cal H}(Z)=\fr12\int\Big[|\Psi|^2+|\na\Lambda|^2\Big]dx+\frac {1}{2I}\langle \Lambda,J\varrho\rangle^2,\quad
Z=(\Lambda,\Psi)\in {\cal E}.
\ee
ii) Energy conservation  holds for solutions $Z(t)\in C^1(\R,\cE)$:
\be\la{enec}
{\cal H}(Z(t))=\const,\qquad t\in\R.
\ee
\end{lemma}
\setcounter{equation}{0}
\section{Solution of equation \eqref{lin}}
We apply the Fourier-Laplace transform
\be\la{FL}
\ti Z(\lam)=\int_0^\infty e^{-\lam t} Z(t)dt,~~~~~~~\rRe\lam>0.
\ee
According to Proposition \ref{wp}, $Z(t)=Y(t)-Y_{\om_*}$ is bounded in $\cE$.
Hence the integral (\ref{FL}) converges and is analytic for $\rRe\lam>0$, and
\be\la{PW}
\Vert\ti Z(\lam)\Vert_{\cE}\le \fr{C}{\rRe\lam},~~~~~~~\rRe\lam>0.
\ee
Applying the Fourier-Laplace transform to (\ref{lin}), we get
\be\la{A-lam}
\lam\ti Z(\lam)={\bf H}\ti Z(\lam)+Z_0, \quad\rRe\lam>0.
\ee
The solution $\ti Z(\lam)$ to \eqref{A-lam} is given by
\be\la{A-lam-sol}
\ti Z(\lam)=-({\bf H}-\lam)^{-1}Z_0,\quad\rRe\lam>0.
\ee
where the resolvent $({\bf H}-\lam)^{-1}$ exists for $\rRe\lam>0$ by the following lemma.
\begin{lemma}
The operator-valued function $({\bf H}-\lam)^{-1}:\cE\to \cE$ is analytic for $\rRe\lam>0$.
\end{lemma}
\begin{proof}
First, the energy conservation \eqref{enec} and  non negativity of  ${\cal H}(Z)$ imply that ${\rm Ker}({\bf H}-\lam)=0$.
Indeed, otherwise there is a nonzero solution $Z(t)=e^{\lam t}Z_\lam$
for which the energy cannot be conserved.
Second, ${\bf H}-\lam={\bf H}_0-\lam +T$, where 
${\bf H}_0=\begin{pmatrix}
0 & 1 \\ \De  & 0 \\ \end{pmatrix}$, and $T$ is a finite-dimensional operator. 
 Operator $({\bf H}_0-\lam)^{-1}$ is bounded in $\cE$ for every $\rRe\lam>0$. Hence,  ${\bf H}-\lam$ is invertible by the Fredholm theory.
 \end{proof}
We rewrite \eqref{A-lam} as
\be\la{F0}
\left\{
\ba{rcl}
\ti\Pi(x,\lam)-\lam\ti\Lam(x,\lam)&=&-\Lam_0(x)
\\
\Delta\ti\Lam(x,\lam)-\lam\ti\Psi(x,\lam)&=&-\Pi_0(x)+\ti\nu(\lam) J\varrho(x)
\ea\right|,\qquad \ti\nu(\lam)=\frac 1I\langle\ti\Lam(\lam),J\varrho\rangle.
\ee
Equivalently, 
in the Fourier space, 
\be\la{F1}
\left\{
\ba{rcl}
\hat{\ti\Pi}(k,\lam)-\lam\hat{\ti \Lam}(k,\lam)&=&-\hat\Lam_0(k)
\\
-k^2\hat{\ti \Lam}(k,\lam)-\lam\hat{\ti \Pi}(k,\lam)&=&-\hat\Pi_0(k)+
\ti\nu(\lam) J\hat\varrho(k)
\ea\right|,\quad k\in\R^2\5.
\ee
We have
$$
\left(
\ba{cc}
-\lam & 1 \\
-k^2 & -\lam
\ea
\right)^{-1}=\frac{1}{k^2+\lam^2}\left(
\ba{cc}
-\lam & -1 \\
k^2 & -\lam
\ea
\right).
$$
Hence,
\beqn\la{hat-Lam}
\hat{\ti\Lam}&=&\frac{\lam\hat\Lam_0+\hat\Pi_0-\ti\nu J\hat\varrho}{k^2+\lam^2}=\frac{\hat K_0}{k^2+\lam^2}-\frac{\ti\nu J\hat\varrho}{k^2+\lam^2},
\qquad \hat K_0(\lam):=\lam\hat\Lam_0+\hat\Pi_0.
\eeqn
Now
$$
I\ti{\nu}(\lam)=\langle \hat{\ti\Lam},J\hat\varrho\rangle=\langle \frac{\hat K_0}{k^2+\lam^2},J\hat\varrho\rangle
-\ti{\nu} \langle \frac{J\hat\varrho}{k^2+\lam^2},J\hat\varrho\rangle
=\langle \frac{\hat K_0}{k^2+\lam^2},J\hat\varrho\rangle-\ti{\nu}\varkappa(\lam),
$$
where $\varkappa(\lam)$ is defined in \eqref{kappa}.
Thus,
\be\la{nu}
\ti{\nu}(\lam)=\frac{\langle \frac{\hat K_0}{k^2+\lam^2},J\hat\varrho\rangle}{I+\varkappa(\lam)}.
\ee
\setcounter{equation}{0}
\section{2D Schr\"odinger resolvent}
Denote by $R(\zeta)=(-\Delta-\zeta)^{-1}$ the resolvent of 2D Laplacian.
Recall  the properties of  $R(\zeta)$ \cite{A, KK2012}:
\\
i) $R(\zeta)$ is strongly analytic function of $\zeta\in\C\setminus[0,\infty)$ with the values in  $H^{-1}(\R^2)\to H^{1}(\R^2)$.
\\
ii) For $\zeta\in(0,\infty)$, the convergence holds 
\be\la{lap}
R^{(k)}(\zeta\pm i\ve)\to R^{(k)}(\zeta\pm i0):=R_{\pm}^{(k)}(\zeta),\quad  \ve\to +0,\quad k=0,1,2
\ee
in  $H^{-1}_{\si}(\R^2) \to H^{1}_{-\si}(\R^2)$  with $\si>k+1/2$.
\\
iii) For  $s=0,1,2$  and $|l|\le 2$,  the asymptotics hold
\be\la{A}
 \Vert R_{\pm}^{(k)}(\zeta^2)\Vert_{H^s_\si\to H^{s+l}_{-\si}}
 ={\cal O}(|\zeta|^{-(1-l)}),\quad \zeta \to \infty, \quad k=0,1,2,\quad \si>1/2+k.
 \ee
Now we obtain asymptotics of $R_{\pm}(\zeta^2)$ for $\zeta\to 0$.  Denote  by $P_{\pm}(\zeta)$  an integral operator with the kernel 
\be\la{P0}
P_{\pm}(\zeta,z):=-\frac{\log (\zeta|z|)}{2\pi} + h_{\pm},
~~{\rm where}~~ h_{\pm}=\pm\frac{i}4+\frac{\gamma}{2\pi}+\frac{\log 2}{2\pi}.
\ee
Here $\gamma$ is Euler's constant.  The following lemma is   an adapted version of \cite[Lemma 5]{S2005}. 
\begin{lemma}\la{t-ek}
The following asymptotics hold
\be\la{as-S}
\Vert \pa^k_{\zeta}\big(R_{\pm}(\zeta^2)-P_{\pm}(\zeta)\big)\Vert_{L^2_\beta\to L^2_{-\beta}}={\cal O}(\zeta^{\frac 32-k}),\quad \zeta\to +0,\quad\beta>5/2,\quad k=0,1,2.
\ee
\end{lemma}
We give  the proof in Appendix \ref{aA} for the sake of completeness.
Denote by  $g(\lam)$ for $\rRe\lam>0$,  the integral operator with the integral kernel 
\be\la{g}
g(\lam,z):=F^{-1}_{k\to z}[\frac{1}{k^2+\lam^2}].
\ee
By \eqref{lap}, there exists
\be\la{g-rep0}
g(i\mu+\ve)\to g(i\mu+0)=R_{\mp}(\mu^2),\quad \ve\to +0, \quad \pm\mu>0
\ee
in  $H^{-1}_{\si}(\R^2) \to H^{1}_{-\si}(\R^2)$  with $\si>1/2$.
The representation  and asymptotics  \eqref{A} imply that 
\be\la{g-as}
 \Vert g^{(k)}(i\mu+0)\Vert_{H^s_\si \to H^{s+l}_{-\si}}
 ={\cal O}(|\mu|^{-1+l}), \quad \mu\to\infty, \quad k=0,1,2,\quad \si>1/2+k
 \ee
for $s=0,1,2$  and $|l|\le 2$.  
Lemma \ref{t-ek} implies the following asymptotics.
 \begin{lemma}\la{Dr-as1}
Let $\rho\in{\cal S}(\R^2)$ satisfy  \eqref{rosym}  and \eqref{zero}. Then   
\be\la{g-mu-as}
g(i\mu+0)\varrho=g_0\varrho+ g_R(\mu)\varrho,\quad \pm\mu>0,
\ee
where $g_0: H^{-1}_{\beta}\to H^1_{-\beta}$, $\beta>1$ is an integral operator with the kernel $g_0(x-y)=-\frac{1}{2\pi}\log|x-y|$,
and 
\be\la{g-R-as}
\Vert g_R^{(k)}(\mu)\Vert_{L^2_\beta\to L^2_{-\beta}}={\cal O}(|\mu|^{\frac 32-k}),\qquad \mu\to 0,\quad k=0,1,2,\quad\beta>5/2.
\ee
 \end{lemma}
\begin{proof}
For  $\rho$ satisfying  \eqref{rosym} and \eqref{zero}, the formulas  \eqref{P0} imply that
\beqn\nonumber
[P_{\pm}(\mu)\varrho](y)&=&-\frac{1}{2\pi}\int \log|x-y|\varrho(x)dx-(\frac{\log |\mu|}{2\pi}-h_{\pm})\int \varrho(x)dx\\
\nonumber
&=&-\frac{1}{2\pi}\int \log|x-y|\varrho(x)dx.
\eeqn
Hence, \eqref{g-mu-as}--\eqref{g-R-as} follows by \eqref{as-S} and \eqref{g-rep0}.
\end{proof}
The lemma and asymptotics \eqref{as-S} imply
\begin{cor}\la{cor-as}
Let $\rho\in{\cal S}(\R^2)$ satisfy  \eqref{rosym}    and $f\in L^2_{\beta}$ with $\beta>5/2$.
Then
\be\la{int-rho-f-as}
\langle\frac{\hat \varrho}{k^2+(i\mu+0)^2}, \hat f\rangle=\langle\frac{\hat\varrho}{k^2}, \hat f\rangle
+{\cal O}(|\mu|^{3/2}),\qquad \mu\to 0.
\ee
The asymptotics can be differentiated  two times.
\end{cor}
\setcounter{equation}{0}
\section{Asymptotics of $\ti\nu$}
 First we  derive   asymptotics of  $\nu(i\mu+0)$ as $\mu\to 0$. 
 \begin{lemma}\la{Lem-q-est}
Let conditions \eqref{rosym}--\eqref{zero} hold, and $Z_0=(\Lam_0,\Psi_0)\in{\cal E}_{\beta}^+$  with $\beta>5/2$.
Then
\be\la{nu-as0}
\ti\nu(i\mu+0)=\Big(-\langle \frac{J\hat\varrho}{k^2},\ov{\hat\Psi}_0\rangle
+i\mu \langle \frac{J\hat\varrho}{k^2},\ov{\hat\Lam}_0\rangle\Big)\frac{1}{I+\varkappa(0)}+{\cal O}(|\mu|^{3/2}), \quad \mu\to 0,
\ee
where $\varkappa(0)=\displaystyle\int\frac{|\hat\varrho(k)|^2dk}{k^2}>0$. The asymptotics allows two differentiations.
\end{lemma}
\begin{proof}
Definition \eqref{kappa} of $\varkappa$ and Corollary \ref{cor-as}  imply that
$$
\varkappa(i\mu+0)=\langle\frac{\hat \varrho}{k^2+(i\mu+0)^2},\hat \varrho\rangle=\varkappa(0)+{\cal O}(\mu^2\log|\mu|), \quad \mu\to 0.
$$
Further,  Corollary \ref{cor-as}  implies that
\beqn\nonumber
\langle \frac{\hat K_0(i\mu+0)}{k^2+(i\mu+0)^2},J\hat\varrho\rangle&=&i\mu\langle\frac{ J\hat\varrho}{k^2+(i\mu+0)^2},\ov{\hat \Lam}_0\rangle
-\langle\frac{J\hat \varrho}{k^2+(i\mu+0)^2},\ov{\hat \Psi}_0\rangle\\
\nonumber
&=&-\langle \frac{J\hat\varrho}{k^2},\ov{\hat\Psi}_0\rangle
+i\mu \langle \frac{J\hat\varrho}{k^2},\ov{\hat\Lam}_0\rangle+{\cal O}(\mu^2\log|\mu|), \quad \mu\to 0.
\eeqn
Hence, \eqref{nu-as0} follows by definition (\ref{nu}) of $\ti\nu$.
\end{proof}
Now we estimate $\ti\nu(i\mu+0)$ for large $\mu\in\R$.
 \begin{lemma}\la{Lem-q-est2}
Let condition \eqref{rosym} hold, and $Z_0=(\Lam_0,\Psi_0)\in{\cal E}_{\beta}^+$  with $\beta>5/2$. Then for sufficiently large $B>0$, the bound holds
\be\la{nu-as}
|\ti\nu^{(k)}(i\mu+0)|\le C|\mu|^{-2}, \quad \mu\in\R, \quad |\mu| \ge B,\quad k=0,1,2.
\ee
\end{lemma}
 \begin{proof}
 Applying \eqref{g-as} with $s=l=0$, $k=0,1,2$ and $\si>1/2+k$, we obtain 
 \beqn\nonumber
|\varkappa^{(k)}(i\mu+0)|&=&|\langle g^{(k)}(i\mu+0)\varrho,\varrho\rangle|\le C\Vert g^{(k)}(i\mu+0)\varrho\Vert_{L^{2}_{-\si}}\Vert\varrho\Vert_{L^2_{\si}}\\
\la{q-as}
&\le& C_1|\mu|^{-1}\Vert\varrho\Vert_{L^2_{\si}}^2\le C_2|\mu|^{-1}, \quad |\mu|\ge B.
\eeqn
Further,  for $\mu\ge B$,  \eqref{g-as} implies 
\beqn\nonumber
|\langle g^{(k)}(i\mu+0)\Psi_0,J\varrho\rangle|&\le& C\Vert g^{(k)}(i\mu+0)\Psi_0\Vert_{H^{-1}_{-\beta}}\Vert \varrho\Vert_{H^1_{\beta}}
\le C_1 |\mu|^{-2}\Vert\Psi_0\Vert_{L^{2}_{\beta}},\\ 
\nonumber
|\langle g^{(k)}(i\mu+0)\Lam_0,J\varrho\rangle| &\le& C\Vert g^{(k)}(i\mu+0)\Lam_0\Vert_{H^{-2}_{-\beta}}\Vert \varrho\Vert_{H^2_{\beta}}
\le C_1 |\mu|^{-3}\Vert\Lam_0\Vert_{L^{2}_{\beta}}, \quad k=0,1,2.
\eeqn
Hence,  
\be\la{K0-as}
|\pa^{k}_{\mu}\langle \frac{\hat K_0}{k^2+(i\mu+0)^2},J\hat\varrho\rangle
\le C|\mu|^{-2}\Vert Z_0\Vert_{{\cal E}_\beta},
\quad |\mu|\ge B, \quad k=0,1,2,
\ee
by definition \eqref{hat-Lam} of $\hat K_0$. Finally, \eqref{nu-as} follow 
from 
 \eqref{nu} and (\ref{q-as})--(\ref{K0-as}).
\end{proof}
\bc\la{cnu}
Lemmas \ref{Lem-q-est} and \ref{Lem-q-est2} together with condition \eqref{M-condition} imply that $\ti\nu(i\mu+0)\in C^2(\R)$.
\ec
\setcounter{equation}{0}
\section{Attraction to soliton in weighted Sobolev norms}
\subsection{Time decay of $\nu(t)$}
Here we prove time decay of $\nu(t)$.
\begin{lemma}\la{vec-decay}
Suppose that the condition \eqref{M-condition} holds, and
$Z_0\in {\cal E}_{\beta}^+$ with $\beta>5/2$. Then $\nu(t)$ is continuous and
\be\la{decQP}
|\nu(t)|\le \frac {C(\rho,\beta)}{(1+t)^2},\qquad t> 0.
\ee
\end{lemma}
\begin{proof}
The component $\nu(t)$ is given by 
$$
\nu(t)=\frac{1}{2\pi}\int e^{i\mu t} \ti\nu(\mu+i0) d\mu=\frac{1}{2\pi}\int_{-\infty}^0 e^{i\mu t} \ti\nu(\mu+i0) d\mu
+\frac{1}{2\pi}\int_0^\infty e^{i\mu t} \ti\nu(\mu+i0) d\mu, \quad t> 0.
$$
Now Corollary \ref{cnu}, asymptotics \eqref{nu-as0}  and bound \eqref{nu-as} 
 imply  \eqref{decQP}  by the double partial integration.
\end{proof}
\subsection{Time decay of $Z(t)$}
\begin{lemma}\la{field-decay}
Suppose that the condition \eqref{M-condition} holds, and
$Z_0\in {\cal E}_{\beta}^+$ with $\beta>5/2$. Then
\be\la{Zdec}
\Vert Z(t)\Vert_{-\beta}\le \frac {C(\rho,\beta)}{(1+t)^2},\qquad t> 0.
\ee
\end{lemma}
\bpr
We rewrite  \eqref{lin} as
\be\la{F-eq}
\dot Z(t)= {\bf H}_0 Z(t)-\begin{bmatrix} 0\\ \nu(t) J\varrho\end{bmatrix},\qquad {\bf H}_0 =\begin{pmatrix}
0 & 1 \\
\De & 0 
\end{pmatrix}.
\ee
To prove time decay of $Z(t)$, we apply the Duhamel representation
\be\la{Duhamel}
Z(t)=W(t)Z_0-\int_0^t W(t-s)\begin{bmatrix} 0\\ \nu(s) J\varrho\end{bmatrix} ds,\quad t\ge 0.
\ee
For the group  $W(t)$ the dispersion decay holds (see Appendix \ref{aB})
\be\la{dede}
\Vert W(t)X\Vert_{-\beta}\le C(1+t)^{-2}\Vert X\Vert_{\beta},\quad \beta> 2,\quad t\ge 0.
\ee
 Applying  \eqref{dede} to \eqref{Duhamel}, and  using  \eqref {decQP}, we obtain \eqref{Zdec}.
 \epr
 \bc
 The  splitting (\ref{YZ}) and the decay (\ref{Zdec}) imply the convergence
 \be\la{Ycon}
 Y(t)\stackrel{\cE_{-\beta}}{-\!\!\!-\!\!\!-\!\!\!\to} Y_{\om_*},\qquad t\to\infty.
 \ee
 \ec
 \section{Asymptotics with dispersive wave in global energy norm}\la{sol-as}
\setcounter{equation}{0}
Here we prove our main Theorem \ref{main}. 
For the considered 
solution $Y(t)$ and for the soliton (\ref{solvom}) with angular velocity $\om_*$, 
\eqref{pM} together with
\eqref{Zdec} imply that
\beqn\nonumber
\om(t)&=&\frac{1}{I}\big(M+\langle A(t),J\varrho\rangle\big)=
\om_*+\frac{1}{I}\big(\langle A(t),J\varrho\rangle-\langle A_{\om_*},J\varrho\rangle\big)\\
\la{om-as}
&=&\om_*+\frac{1}{I}\langle \Lam(t),J\varrho\rangle=\om_*+ {\cal O}(t^{-2}),\quad t\to\infty.
\eeqn
Further, for the difference 
$Z(t) =Y(t) -Y_{\om_*}(t)$,   equation  \eqref{mls3} implies 
$$
\dot Z(t)={\bf A}_0Z(x.t)-R(t),\qquad R(t):=\begin{pmatrix}
0\\\ (\om(t)-\om_*) J\varrho(x)
\end{pmatrix}.
$$
Then similarly to (\ref{Duhamel}),
\beqn\nonumber
Z(t)&=&W(t)Z_0-\int_0^{t}W(t-s)R(s)ds\\
\nonumber
&=&W(t)\Big(Z_0-\int_0^{\infty}W(-s)R(s)ds\Big)+\int_t^{\infty}W(t-s)R(s)ds=W(t)\Psi_++r_+(t).
\eeqn
Here $\Psi_{+} \in{\cal E}$ and  $\Vert r_{+}(t)\Vert_{\cal E}= {\cal O}(t^{-1})$
by (\ref{om-as})
since $W(t)$ is the unitary group  
  in the space $\cE$. 
\appendix
\setcounter{equation}{0}
\section{Proof of Lemma \ref{t-ek}}\la{aA}
Recall that  (cf. \cite[Section 2]{S2005})
\be\la{HJY}
R_{\pm}(\zeta^2,z)=\pm\frac{i}{4} H_0^{\pm}(\zeta|z|),
\ee
where  for $H_0(s)$ the asymptotics hold (see \cite{O})
$$
H_0^{\pm}(s)=1\pm \frac{2i}{\pi}\big(\log \frac{s}2+\gamma\big)+{\cal O}(s^{2}\log s),\quad s\to 0.
$$
The asymptotics imply that
\be\la{g-as-s}
\pa_{\zeta}^kR_{\pm}(\zeta^2,z)=\pa_{\zeta}^k P_{\pm}(\zeta,z)+{\cal O}\big((\zeta|z|)^{2-k}\log (\zeta |z|)\big),\quad \zeta |z|\to 0,\quad k=0,1,2,
\ee
where $P_{\pm}(\zeta,z)=-\frac{\log (\zeta|z|)}{2\pi} + h_{\pm}$ as in \eqref{P0}. Moreover, 
\be\la{g-as-s1}
|\pa_{\zeta}^k R_{\pm}(\zeta^2,z)|=\frac 14|\pa_{\zeta}^k H_0^{\pm}(\zeta|z|)|\le \frac{C|z|^k}{\sqrt{\zeta |z|}}, \quad \zeta |z|\ge 1,\quad k=0,1,2.
\ee
Let  $\chi_\zeta(z)$ be characteristic function of the ball $ |z|\le \frac 1{\zeta}$. We have
\beqn\nonumber
R_{\pm}(\zeta^2,z)-P_{\pm}(\zeta,z)&=&\chi_\zeta(z)\big(R_{\pm}(\zeta^2,z)-P_{\pm}(\zeta,z)\big)+
(1-\chi_\zeta(z))\big(R_{\pm}(\zeta^2,z)-P_{\pm}(\zeta,z)\big)\\
\la{sp}
&=&Q_{\pm}(\zeta,z)+S_{\pm}(\zeta,z).
\eeqn
For $\zeta|z|\le 1$, \eqref{g-as-s} implies
\beqn\nonumber
|Q_{\pm}(\zeta,z)|&\le& C(\zeta|z|)^{2}\log (\zeta |z|)\le C(\zeta|z|)^{3/2},\\
\nonumber
 |\pa_{\zeta} Q_{\pm}(\zeta,z)|&\le& C|z|\zeta|z|\log (\zeta |z|)\le C|z|(\zeta|z|)^{1/2}= C\zeta^{1/2}|z|^{3/2}, \\
\la{Q}
|\pa^2_{\zeta} Q_{\pm}(\zeta,z)|&\le& C|z|^2\log (\zeta |z|)\le C\frac{|z|^2}{\sqrt{\zeta|z|}}= C\zeta^{-1/2}|z|^{3/2}.
\eeqn
 Further, for $\zeta|z|\ge 1$,  \eqref{P0}, \eqref{HJY} and \eqref{g-as-s1}  imply 
\beqn\nonumber
|S_{\pm}(\zeta,z)|&\le & C\Big( 1+\log (\zeta|z|)+\frac{1}{\sqrt{\zeta |z|}}\Big)\le C_1(\zeta |z|)^{3/2},\\
\nonumber
|\pa_{\zeta} S_{\pm}(\zeta,z)|&\le& C\Big(\frac{|z|}{\zeta|z|}+\frac{C|z|}{\sqrt{\zeta |z|}}\Big) \le C |z|(\zeta|z|)^{1/2}=C\zeta^{1/2} |z|^{3/2},\\
\la{S}
|\pa^2_{\zeta} S_{\pm}(\zeta,z)|&\le& C\Big(\frac{|z|^2}{(\zeta|z|)^2}+\frac{C|z|^2}{\sqrt{\zeta |z|}}\Big) \le C\frac{|z|^2}{\sqrt{\zeta|z|}}= C\zeta^{-1/2}|z|^{3/2} .
\eeqn
Finally, \eqref{Q} and \eqref{S}  imply \eqref{as-S} by \eqref{sp}.
\section{Dispersive decay}\la{aB}
\setcounter{equation}{0}
Here we prove the decay \eqref{dede}. Note that 
\be\la{Wv}
 W(z,t)=\left( \ba{ll}
  \dot G(z,t)    &              G(z,t)\\
  \ddot G(z,t)   &        \dot G(z,t)
  \ea \right),\quad G(z,t)=\frac 1{2\pi}\frac{\theta(t-|z|)}{\sqrt{t^2-|z|^2}}, \quad z\in\R^2,\quad t>0.
 \ee 
 {\it Step i)}
  For any $\ve\in (0,1)$, the bounds hold
\beqn\la{G2R}
|\partial^{\alpha}_z\partial^{j}_tG(z,t)|&\le& C(\ve )t^{-2},
\quad1\le|\alpha|+j\le 2, \quad|z|\le\ve t,\quad t\ge 1.
\eeqn
Formula \eqref{Wv} imply  the Huygen's principle  for the group $W(t)$: $W(z, t) = 0$ for  $|z| > t$.
 The energy conservation gives that
\be\la{en-con}
\Vert W(t)Z\Vert_\cE =\Vert Z\Vert_\cE, \quad t\ge 0,\qquad Z\in\cE. 
\ee 
 {\it Step ii)}
We fix an arbitrary $\ve\in (0,1)$.  For any  $t\ge 1$ we
split the  function $Z(x)$ in two terms, $Z=Z_{1,t}+Z_{2,t}$ such that
\be\label{F}
  Z_{1,t}(x)=0~~\mbox{for}~|x|>\ve t/2,
  ~~~~~~~\mbox{and}~~~~~~~~~~
  Z_{2,t}(x)=0~~\mbox{for}~|x|<\ve t/4,
  \ee
with
\begin{equation} \label{FFn}
  \Vert Z_{1,t}\Vert _\cE+\Vert Z_{2,t}\Vert _\cE\le C\Vert Z\Vert _\cE,\quad t\ge 1.
\end{equation}
\textit{Step iii)}
Using \eqref{en-con}--(\ref{FFn}), we obtain for any $\si\ge 0$
\be\label{WF2} 
\Vert W(t)Z_{2,t}\Vert _{\cF _{-\si}}\le \Vert W(t) Z_{2,t}\Vert _{\cF_0}=\Vert Z_{2,t}\Vert _{\cF_0}
 \le  C(\ve) t^{-\si}\Vert Z_{2,t}\Vert _{\cF _\si} \le C(\ve) t^{-\si}\Vert Z\Vert _{\cF _\si},\,\,\, t\ge 1.
\ee
\textit{Step iv)}
Denote by  $\zeta $  the operator of multiplication by the function $\zeta ({|x|}/{t})$ such that $\zeta =\zeta (s)\in C_{0}^{\infty}(\R)$,
$\zeta (s)=1$ for $|s|<\ve/4$, $\zeta (s)=0$ for $ |s|>\ve/2$. Obviously, 
$|\partial _{x}^{\alpha }\zeta ({|x|}/{t})|\leq C(\ve)<\infty$ for $|\alpha|\le 1$ and $t\ge 1.$
Furthermore, $1-\zeta ({|x|}/{t})=0$ for $|x|<\ve t/4$. Hence,  using \eqref{en-con}, we get for $\si\geq 0$
\be \la{f0}
 ||(1-\zeta )W(t) Z_{1,t}||_{\cF _{-\si}}\le C(\ve)t^{-\si}\Vert W(t) Z_{1,t}\Vert_{\cF_0}
\le  C_1(\ve)t^{-\si}\Vert Z\Vert _{\cF _\si},\quad t\ge 1.
 \ee
\textit{Step v)}
It remains to estimate $\zeta W(t) Z_{1,t}$.
Let $\chi _{\ve t/2}$ 
be the characteristic function of the ball $ |x|\le \ve t/2$.
We will use the same notation for the operator of multiplication by this characteristic function. 
By (\ref{F}), 
\be\la{xx}
  \zeta W(t) Z_{1,t}=\zeta W(t)\chi_{\ve t/2}Z_{1,t}.
\ee
The norm of the operator
$\zeta W(t)\chi_{\ve t/2}: \cF _{\si}\rightarrow \cF _{-\si}$ is equivalent to the norm of the operator
$\langle x\rangle^{-\si}\zeta W(t)\chi _{\ve t/2}(y)\langle y\rangle^{-\si}:\cF_0 \rightarrow \cF_0$.
The norm of the later operator does not exceed the sum in $\al$ and $j$, $1\le|\al|+j\le 2$
of the norms of integral operators
\begin{equation}\label{1234}
 \langle x\rangle^{-\si}\pa_t^j \pa_{x}^{\al}[\zeta W(t)\chi_{\ve t/2}(y)\langle y\rangle^{-\si}]: L^2(\R^2)\oplus L^2(\R^2)
\to L^2(\R^2)\oplus L^{2}(\R^2).
\end{equation}
The estimates (\ref{G2R}) imply that these operators  are of the 
Hilbert-Schmidt type
for $\si>2$, and their Hilbert-Schmidt norms do not exceed
$C t^{-2}$. Hence, (\ref{FFn}) and (\ref{xx}) imply that for $\si> 2$
\begin{equation} \label{HS}
  \Vert \zeta W (t)F_{1,t}\Vert_{\cF _{-\si}}\le C(\ve) t^{-2}\Vert F_{1,t}\Vert_{\cF_\si}\le C(\ve) t^{-2}\Vert F_{0}\Vert_{\cF _\si},\quad t\ge 1.
\end{equation}
Finally, the estimates (\ref{WF2}),  (\ref{f0}) and (\ref {HS})  with $\si=\beta$ imply \eqref{dede}.

\end{document}